\documentclass[10pt, hidelinks]{article}
\usepackage{stix}
\usepackage[left=20mm,right=20mm,top=20mm,bottom=25mm,a4paper]{geometry}
\usepackage[utf8]{inputenc}
\usepackage{amsmath}
\DeclareMathOperator{\diag}{diag}

\newcommand{\bC}{\mathbf{C}}
\newcommand{\bF}{\mathbf{F}}
\newcommand{\bI}{\mathbf{I}}
\newcommand{\bH}{\mathbf{H}}
\newcommand{\bnull}{\mathbf{0}}
\newcommand{\SP}{\mathcal{SP}}
\newcommand{\SM}{\mathcal{SM}}
\newcommand{\bsigma}{\boldsymbol{\upsigma}}
\newcommand{\btau}{\boldsymbol{\uptau}}
\newcommand{\ba}{\boldsymbol{a}}
\newcommand{\bb}{\boldsymbol{b}}
\newcommand{\be}{\boldsymbol{e}}

\newcommand{\SL}{\operatorname{SL}}
\newcommand{\SYM}{\text{Sym}}

\newcommand{\tr}{\operatorname{tr}}
\newcommand{\cof}{\operatorname{Cof}}
\newcommand{\bbR}{{\mathbb{R}}}
\newcommand{\cC}{{\mathcal{C}}}
\newcommand{\cP}{{\mathcal{P}}}
\newcommand{\bV}{\mathbf{V}}

\usepackage{enumitem}

\usepackage{amsmath,amsthm, upgreek}
\usepackage{centernot}
\counterwithin{equation}{section}
\usepackage{nicefrac,bbm}
\usepackage{bm}
\usepackage{cancel}
\theoremstyle{definition}   %   keine kursive Schrift in der "theorem" - Umgebung
\usepackage{apptools,mathtools,chngcntr}
\usepackage[dvipsnames,svgnames]{xcolor}
\usepackage{tikz}
\usepackage{pgfplots}
\usepgfplotslibrary{groupplots}
\usepackage{pgfplotstable}
\usepackage{subcaption}
\usepackage[font=small]{caption}
\usepackage{wrapfig}
\pgfplotsset{compat=newest}
\usepgfplotslibrary{fillbetween}
\usetikzlibrary{patterns}
\usepgfplotslibrary{colormaps}
\usetikzlibrary{arrows.meta}

\usepackage[maxbibnames=50, maxcitenames=2, uniquelist=false,
    natbib=true, style=authoryear-comp,
    isbn=false, url=false,giveninits=true,dashed=false]{biblatex}   % Literaturverzeichnis
\bibliography{bibliography}
\AtBeginBibliography{} % no comma before last author in bibliography
\usepackage[ngerman, main=english]{babel}
\usepackage[autostyle]{csquotes}% Anführungszeichen vereinfacht
\usepackage[T1]{fontenc}
\usepackage{hyperref}
\usepackage[affil-it]{authblk}
\usepackage[capitalise]{cleveref}
\crefname{section}{Sect.}{Sect.}
\crefname{table}{Tab.}{Tab.}

\newtheorem{definition}{Definition}[section]
\newtheorem{theorem}[definition]{Theorem}%[section]
\newtheorem{corollary}[definition]{Corollary}%[section]

\definecolor{changes}{RGB}{192, 0, 0}

\title{
Polyconvexity does not imply true-stress-true-strain monotonicity\\ in the incompressible three-dimensional case
}
\bigbreak
\author[1,*]{Dominik~K.~Klein}
\author[2]{Maximilian~P.~Wollner}
\author[3]{Patrizio~Neff}
\affil[1]{\footnotesize Cyber-Physical Simulation, Department of Mechanical Engineering, TU Darmstadt, 64293 Darmstadt, Germany}
\affil[2]{\footnotesize Institute of Biomechanics, Graz University of Technology, Stremayrgasse 16/2, 8010, Graz, Austria}
\affil[3]{\footnotesize Chair of Nonlinear Analysis and Modelling, University of Duisburg-Essen, Thea-Leymann-Straße 9, 45127, Essen, Germany}
\affil[*]{\footnotesize Corresponding author. E-mail addresses: klein@cps.tu-darmstadt.de, wollner@tugraz.at, patrizio.neff@uni-due.de}
\date{June 10, 2026}
\begin{document}
\maketitle
\par\noindent\rule{\textwidth}{0.4pt}
\begin{abstract}
\noindent 
We study constitutive conditions of hyperelastic potentials for incompressible material behavior in three dimensions. By means of a counterexample, we show that polyconvexity does not imply true-stress-true-strain monotonicity. Thus, polyconvexity alone is not strong enough to guarantee a physically reasonable response for idealized elasticity.
\end{abstract}
\bigbreak

\par\noindent\rule{\textwidth}{0.4pt}
%
%========================================================================

% \linenumbers

%%%%%%%%%%%%%
%	Introduction
%%%%%%%%%%%%%
\section{Introduction}

We consider incompressible hyperelastic potentials $W$ and the associated Cauchy stress tensor $\bsigma$ given by
\begin{equation}
W:\SL(3)\times\bbR\rightarrow\bbR\,,\quad (\bF,\,p)\mapsto \widetilde{W}(\bF)-p\,(J-1) \qquad \text{and}\qquad \bsigma=\mathrm{D}_{\bF}\widetilde{W}\,\bF^T-p\bI\,,
\end{equation}
where $\bF=\mathrm{D}\boldsymbol{\varphi}$ is the deformation gradient derived from the deformation $\boldsymbol{\varphi}$, $p$ is a Lagrange multiplier ensuring incompressibility ($J=\det\bF=1$), and $\bI$ is the second-order identity tensor. Here, $\operatorname{SL}(3):=\big\{\bF \in\allowbreak \;\bbR^{3\times 3}\,\rvert\,\allowbreak \det \bF =1\big\}$ denotes the special linear group in three dimensions. We furthermore assume isotropy and $W\in\cC^2$. One of the main challenges in material theory is the formulation of constitutive constraints imposed on $W$ and the investigation on how they relate to each other. Setting aside more obvious requirements such as frame indifference, we focus on the following conditions:
\begin{definition}
 The potential $W$ is polyconvex if there exists a representation
    \begin{equation}
        W(\bF)=\cP(\bF,\,\cof\bF)\,,
    \end{equation}
    where $\cP:\bbR^{3\times 3}\times \bbR^{3\times 3}\rightarrow\bbR$ is a convex function. Note that the representation of polyconvex functions is non-unique.
\end{definition}
\begin{definition}
The potential $W$ is rank-one convex if it satisfies
\begin{equation}\label{eq:rank_one}
 W(\bF+t\,\ba\otimes\bb)\leq t\, W(\bF+\ba\otimes\bb)+(1-t)\,W(\bF)\quad \forall\,t\in[0,1],\,\ba,\bb\in\bbR^3 \quad \text{with} \quad \ba\otimes\bb\in\operatorname{T}_{\operatorname{SL}(3)}(\bF)\,,
\end{equation}
where
\begin{equation}
  \operatorname{T}_{\operatorname{SL}(3)}(\bF):=\big\{\bH \in\allowbreak \;\mathbb{R}^{3\times 3}\,\rvert\,\allowbreak \bigl\langle\bH,\bF^{-T}\bigr\rangle=0\big\}
\end{equation} 
is the tangent space to $\SL(3)$ at $\bF$ \parencite{Dunn2003}, which is equivalent to $\bF+\ba\otimes\bb\in\operatorname{SL}(3)$ \parencite[App.~A]{Klein2026a}. Given sufficient differentiability, \eqref{eq:rank_one} implies the Legendre-Hadamard (LH) ellipticity condition
      \begin{equation}
  \bigl\langle \mathrm{D}^2_{\bF} W(\bF).(\ba\otimes\bb),\ba\otimes\bb \bigr\rangle\geq 0 \quad \forall\,\ba,\bb\in\bbR^3\,,\,\,\,\ba\otimes\bb\in\operatorname{T}_{\operatorname{SL}(3)}(\bF)\,,
      \end{equation}
where $\langle(\bullet),(\bullet)\rangle$ denotes the inner product between two tensors of equal order and the dot in $(\bullet).(\bullet)$ denotes the double contraction of a fourth-order tensor with a second-order tensor resulting in a second-order tensor.
\end{definition}
\begin{definition}
The potential $W$ fulfills the {true-stress-true-strain monotonicity} (TSTS-M) condition if the Cauchy stress satisfies
    \begin{equation}
        \bigl\langle  \bsigma (\log\bV_{2})- \bsigma(\log\bV_{1}),\log\bV_{2} - \log\bV_{1}\bigr\rangle > 0\qquad\forall\,\bV_1,\bV_2\in\SYM^{++}(3)\,,\,\,\,\bV_{1} \neq \bV_{2}\,,
    \end{equation}
    where $\bV=\sqrt{\bF\,\bF^T}$ is the right stretch tensor and $\log\bV$ denotes the Hencky strain \parencite{Neff2025c,Neff2016a}. Here, $\SYM^{++}(3):=\big\{\bV \in \bbR^{3\times 3}\,\rvert\,\allowbreak \bV=\bV^T,\,\langle\ba,\bV\ba\rangle>0\:\forall\,\ba\in\bbR^3\setminus\{\bnull\}\big\}$ denotes the space of positive definite second-order tensors. Moreover, for incompressibility, the TSTS-M condition is equivalent to the inequality proposed by \textcite{Hill1970} reading
    \begin{equation}\label{eq:hill}
        \bigl\langle  \btau(\log\bV_{2}) - \btau(\log\bV_{1}),\log\bV_{2} - \log\bV_{1}\bigr\rangle > 0\qquad\forall\,\bV_1,\bV_2\in\SYM^{++}(3)\,,\,\,\,\bV_{1} \neq \bV_{2}\,,\,\,\,\tr\left(\log{\bV_i}\right)=0\,,
    \end{equation}
with the Kirchhoff stress \parencite{neff2019axiomaticintroductionarbitrarystrain,richter1949}
\begin{equation}
    \btau=\mathrm{D}_{\log\bV}\widehat{W}(\log\bV)\quad\text{with}\quad \widehat{W}(\log\bV) = W(\bF)\,,
\end{equation}
since $\btau=\bsigma$ for $\det\bF=1$ \parencite{Wollner2026b,Baaser2026}. Making use of the reduced representation 
\begin{equation}\label{eq:pot_reduced}
\widehat{W}_{\text{red}}^{\text{inc}}(\log\lambda_1,\log\lambda_2) := \Psi(\lambda_1,\lambda_2,1/(\lambda_1\lambda_2)) = \Psi(\lambda_1,\lambda_2,\lambda_3)= W(\bF)\,,
\end{equation}
convexity of $\widehat{W}_{\text{red}}^{\text{inc}}$ in $(\log\lambda_1,\log\lambda_2)$ can be derived as a sufficient and necessary condition for Hill's inequality in the incompressible case \parencite[App.~D.3]{Baaser2026}.
\end{definition}
Polyconvexity is linked to existence theorems in finite elasticity theory \parencite{Ball1976,Ball1977}. LH-ellipticity, from a physical perspective, guarantees the existence of real-valued wave speeds for solutions of the governing equations of finite elasticity theory \parencite{Zee1983}. Moreover, it is commonly employed to ensure a stable and robust behavior when applying a constitutive model in numerical applications such as the finite element method \parencite{Schroeder_Neff_Balzani_2005}, see also \parencite[Sect.~1.1]{klein2026limitationspolyconvexity}. On the other hand, the TSTS-M condition can be seen as the multiaxial analogue of the notion that the Cauchy stress must be increasing with increasing Hencky strain, which seems a reasonable requirement for a purely elastic material law. Moreover, TSTS-M ensures that the tangent operator in a hypoelastic formulation is positive definite, i.e.,
\begin{equation}
  \bigl\langle \mathrm{D}_{\log\bV} \bsigma(\log\bV).\bH,\bH \bigr\rangle > 0 \quad \forall\,\bH\in\bbR^{3\times 3}\setminus\{\bnull\} \,,
\end{equation}
and it supports a local existence proof of rate-form equilibrium, even without LH-ellipticity \parencite{Neff2026a,NEFF2025105033}.

  %	overview
 \begin{figure}[t!]
    \centering
    \includegraphics{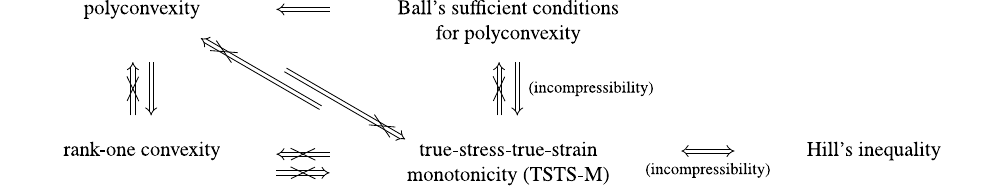}
    \caption{Overview of various constitutive constraints and their relation in isotropic incompressible hyperelasticity.}
    \label{fig:overview}
 \end{figure}
 
In the following, we consider two parameterizations for polyconvex potentials:
 \begin{theorem}\label{th:ball}
      \textcite[Thm.~5.1]{Ball1976} provides \emph{sufficient but not necessary} conditions for polyconvexity. For this, we consider
    \begin{equation}
W(\bF)=g(\boldsymbol{\varkappa}_\lambda)\quad \text{with}\quad \boldsymbol{\varkappa}_\lambda=(\lambda_1,\,\lambda_2,\,\lambda_3,\,\lambda_2\lambda_3,\,\lambda_1\lambda_3,\,\lambda_1\lambda_2)\quad \text{and}\quad\widetilde{g}(\lambda_1,\,\lambda_2,\,\lambda_3)=g(\boldsymbol{\varkappa}_\lambda)\,,
    \end{equation}
where the principal stretches $\lambda_i$ are the singular values of $\bF$ and $\lambda_i\lambda_j$ are the singular values of $\cof\bF$. Polyconvexity is fulfilled if {(i)} $g$ is convex and monotonically increasing and {(ii)} $\widetilde{g}$ satisfies the permutation invariance\footnote{For Ball's original proof, a stricter permutation invariance is required, i.e., it must be independently hold for the first three and the second three arguments of $g$. However, \textcite[Sect. 3.3]{Geuken2026} recently showed that the relaxed permutation invariance in~\eqref{eq:perm_inv_Ball} is also sufficient.}
\begin{equation}\label{eq:perm_inv_Ball}
    \widetilde{g}(\lambda_1,\,\lambda_2,\,\lambda_3)=\widetilde{g}(\lambda_1,\,\lambda_3,\,\lambda_2)=\widetilde{g}(\lambda_2,\,\lambda_1,\,\lambda_3)=\widetilde{g}(\lambda_2,\,\lambda_3,\,\lambda_1)=\widetilde{g}(\lambda_3,\,\lambda_1,\,\lambda_2)=\widetilde{g}(\lambda_3,\,\lambda_2,\,\lambda_1)\,.
\end{equation} 
 \end{theorem}
 \begin{theorem}\label{th:SSV}
\textcite{Wiedemann2026} provide \emph{sufficient and necessary} conditions for polyconvexity in terms of the signed singular values $\nu_i$ of $\bF$ and the signed singular values $\nu_i\nu_j$ of $\cof\bF$.\footnote{Sufficient and necessary conditions for polyconvexity of isotropic potentials were also established in earlier works \parencite{Mielke2005, Rosakis1997}, for instance, \textcite{Dacorogna1993} proposed conditions similar to Theorem~\ref{th:SSV} for the two-dimensional case, referred to as \emph{diagonal polyconvexity}.} For this, we consider
        \begin{equation}\label{eq:pot_SSV}
W(\bF)=h(\boldsymbol{\varkappa}_\nu)\quad \text{with}\quad \boldsymbol{\varkappa}_\nu=(\nu_1,\,\nu_2,\,\nu_3,\,\nu_2\nu_3,\,\nu_1\nu_3,\,\nu_1\nu_2)\quad \text{and}\quad\widetilde{h}(\nu_1,\,\nu_2,\,\nu_3)=h(\boldsymbol{\varkappa}_\nu)\,.
    \end{equation}
Polyconvexity is fulfilled if and only if (i) $h$ is convex, (ii) $\widetilde{h}$ is $\Pi_3$-invariant, which includes the six permutations introduced in~\eqref{eq:perm_inv_Ball} and the four symmetries
\begin{equation}\label{eq:SSV_symmetries}
\begin{aligned}
    \widetilde{h} (\nu_1,\,\nu_2,\,\nu_3)=\widetilde{h} (-\nu_1,\,-\nu_2,\,\nu_3)=\widetilde{h} (-\nu_1,\,\nu_2,\,-\nu_3)=\widetilde{h} (\nu_1,\,-\nu_2,\,-\nu_3)\,,
\end{aligned}
\end{equation}
and (iii) $h$ is lower semi-continuous.
 \end{theorem}

The following relations are well-established for incompressible hyperelasticity in three dimensions: 
\begin{itemize}[label=\textbf{--}]
    \setlength{\itemsep}{0pt}
\item Rank-one convexity does not imply polyconvexity \parencite{Ciarlet1988}.
\item TSTS-M does not imply polyconvexity or rank-one convexity \parencite[Sect.~3.3]{Wollner2026b}.
\item Polyconvexity implies rank-one convexity \parencite{Ciarlet1988}. Moreover, we recently showed that the polyconvex parameterization by Ball, cf.\ Theorem~\ref{th:ball}, implies TSTS-M \parencite[Sect.~3.1.4]{Wollner2026b}.
\end{itemize}

%	softplus
\begin{figure}[t!]
    \centering            \includegraphics[width=0.6667\textwidth]{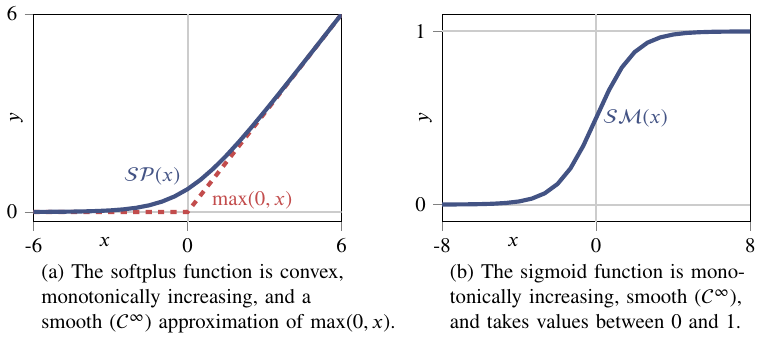}
 \caption{Visualization of the softplus function $\SP(x)$ employed for the potential in \eqref{eq: strain-energy function}. The stress response of the potential, cf.\ \eqref{eq:tau_model} and \eqref{eq:cauchy_model}, depends on the sigmoid function $\SM(x)=\mathrm{D}_x\SP(x)$, which is the first derivative of the softplus function.}
         \label{fig:softplus}
\end{figure}

Above introduced conditions are of particular interest for one of the main open problems of material theory, often referred to as Truesdell's Hauptproblem \parencite{Truesdell1956}: If we disregard any inelastic effects such as fatigue, softening, or plasticity, what set of constitutive constraints is required to represent idealized elasticity? Setting aside obvious requirements such as frame indifference, one approach could be the combination of polyconvexity (and thus LH-ellipticity) in combination with TSTS-M \parencite{Wollner2026a,Neff2025c}. Together, these ensure a monotonically increasing Cauchy stress response for load scenarios where we would expect such a behavior, notably for uniaxial tension, equibiaxial tension, and simple shear. While additional constraints might be required to represent idealized elasticity\footnote{We recently raised the question whether for idealized elasticity, in addition to TSTS-M which restricts the slope of the Cauchy stress, an additional constraint on the curvature of the Cauchy stress might be required \parencite[Sect.~5.4]{Wollner2026b}.}, the aforementioned conditions already go a long way towards this goal. In incompressibility, we recently showed that Ball's sufficient conditions for polyconvexity also imply TSTS-M \parencite[Sect.~3.1.4]{Wollner2026b}, which applies to a variety of models based on principal stretches and models based on the main invariants of the right Cauchy-Green tensor $\bC = \bF^T\bF$. However, the parameterization of Ball is only sufficient but not necessary for polyconvexity. Therefore, in the incompressible three-dimensional case, it remains unclear whether polyconvexity in general implies TSTS-M.\footnote{In the three-dimensional compressible case, polyconvexity does not imply TSTS-M \parencite{Wollner2026a}, while in the incompressible two-dimensional case, polyconvexity does imply TSTS-M \parencite{Ghiba2026b}. However, the transition between two and three dimensions or compressibility and incompressibility entails some theoretical nuances, requiring independent investigations for each setting. Thus, it remained an open problem whether polyconvexity or rank-one convexity imply TSTS-M in the incompressible three-dimensional case.} In this contribution, we address this and show the following: 
\begin{itemize}[label=\textbf{--}]
    \setlength{\itemsep}{0pt}
\item Polyconvexity does not imply TSTS-M in the incompressible three-dimensional case.
\end{itemize}
For this, we provide a counterexample which is polyconvex but does not satisfy TSTS-M, making use of the sufficient and necessary conditions for polyconvexity proposed by \textcite{Wiedemann2026}.\footnote{The construction of such a counterexample was part of a series of challenges posed by Patrizio Neff. In particular, he still offers a prize money of 500€ for the construction of a compressible energy that simultaneously satisfies polyconvexity (or rank-one convexity) and TSTS-M for all deformation states \parencite{Neff2025}.} Since polyconvexity implies rank-one convexity, our counterexample also establishes the following result:
\begin{itemize}[label=\textbf{--}]
\setlength{\itemsep}{0pt}
\item Rank-one convexity does not imply TSTS-M in the incompressible three-dimensional case. 
\end{itemize}
Overall, our work completes the investigation of the relationships between the above introduced constitutive conditions, cf.\ Fig.\ \ref{fig:overview}.

%	uniax example
\begin{figure}[t!]
    \centering            \includegraphics[width=\textwidth]{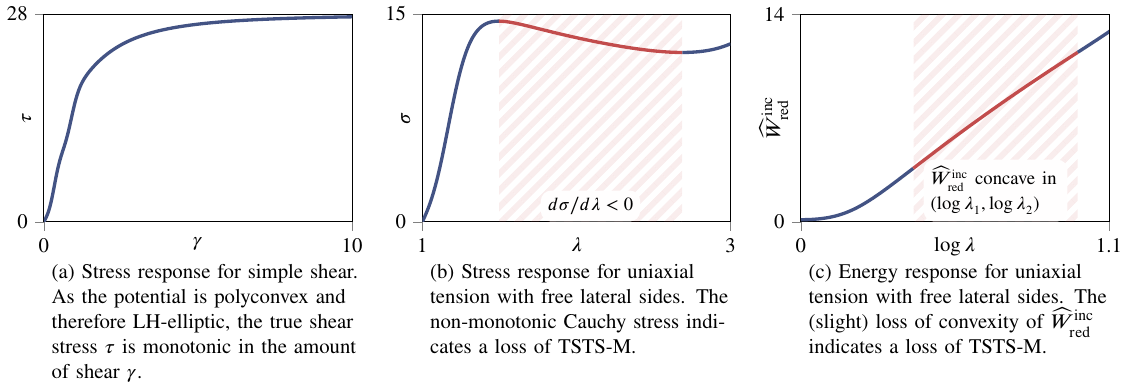}
 \caption{Evaluation of the potential in~\eqref{eq: strain-energy function} with parameters $  (a,\,b,\,c)=(-5,\,-14,\,-22)$. The potential is polyconvex but violates TSTS-M \parencite{Neff2025b}. This counterexample shows that polyconvexity does not imply TSTS-M in the incompressible three-dimensional case.}
         \label{fig:uniax_example}
\end{figure}

%%%%%%%%%%%%%
%	A counterexample
%%%%%%%%%%%%%
\section{A counterexample}

Let us consider the isotropic incompressible potential
\begin{equation}
\label{eq: strain-energy function}
\begin{aligned}
    \psi(\nu_1, \nu_2, \nu_3)= \phi(\nu_1, \nu_2, \nu_3) + \phi(-\nu_1, -\nu_2, \nu_3) + \phi(\nu_1, -\nu_2, -\nu_3) + \phi(-\nu_1, \nu_2, -\nu_3) + \text{const.}\,,
\end{aligned}
\end{equation}
based on the softplus function $\SP(x) = \log(1 + \exp x)$, which we visualize in Fig.~\ref{fig:softplus}(a), with
\begin{equation}\label{eq:pot_one}
 \phi(\nu_1, \nu_2, \nu_3)= \SP(\theta(\boldsymbol{\varkappa}_\nu))\quad\text{and}\quad\theta(\boldsymbol{\varkappa}_\nu)=a\,(\nu_1 + \nu_2 + \nu_3) + b\,(\nu_1\nu_2 + \nu_2\nu_3 + \nu_3\nu_1) + c\,,
\end{equation}
where $a,b,c \in \mathbb{R}$ are material parameters.\footnote{The structure of the potential in \eqref{eq: strain-energy function} closely resembles constitutive models based on neural networks, which are often based on linear transformations and the softplus function \parencite{Wollner2026b}.} 

\begin{theorem}\label{th:pc}
The potential $\psi$ in~\eqref{eq: strain-energy function} is polyconvex for all choices of $a,b,c\in\bbR$.
\end{theorem}
\begin{proof}
The potential $\psi$ satisfies the following sufficient and necessary conditions for polyconvexity, cf.\ \textcite[Cor.\ 2]{Geuken2026} and Theorem \ref{th:SSV}:
\begin{enumerate}
    \item[(i)] The potential $\psi$ is convex in $\boldsymbol{\varkappa}_\nu=(\nu_1,\,\nu_2,\,\nu_3,\,\nu_2\nu_3,\,\nu_1\nu_3,\,\nu_1\nu_2)$ for any choice of $a,b,c\in\bbR$. This can be seen as follows: The function $\theta$, defined in \eqref{eq:pot_one}, is linear in $\boldsymbol{\varkappa}_\nu$. Since the softplus function $\SP(x) = \log(1 + \exp x)$ is convex, it follows that $\SP(\theta(\boldsymbol{\varkappa}_\nu))$ is also convex. Then, notice that $\psi$ is a simple sum of functions of the form $\SP(\theta(\boldsymbol{\varkappa}_\nu))$, albeit with a change of sign in the arguments. But these sign reversals leave the linearity of $\theta$ untouched, hence each summand remains convex following the argument above.
    \item[(ii)] The function $\psi$ obeys $\Pi(3)$-invariance, since it is invariant under permutation of its arguments and satisfies the symmetries introduced in \eqref{eq:SSV_symmetries} by construction.
    \item[(iii)] The strain-energy function $\psi$ is continuous and therefore lower semi-continuous.
\end{enumerate}
\end{proof}
An alternative, more direct proof can be achieved through \textcite[Prop.\ 5.2]{Wiedemann2026}, cf.\ also \textcite[Sect.\ 3.2.2]{Geuken2026}.

As a sanity check, we may verify the monotonicity of the true-shear-stress response given the potential in \eqref{eq: strain-energy function} in simple shear, as required by rank-one convexity. For simple shear along $\be_1$-$\be_2$ in an underlying Cartesian coordinate system, the deformation gradient $\bF_\mathrm{ss}$ and the true shear stress $\tau$ are given by
\begin{equation}
[\bF_\mathrm{ss}]=\bI+\gamma\,\be_1\otimes\be_2\quad\text{and}\quad\tau=\bigl\langle\bsigma_\mathrm{ss}, \be_1\otimes\be_2\big\rangle\,,
\end{equation}
where $\gamma\in\bbR$ is the amount of shear and $\bsigma_\mathrm{ss}$ denotes the Cauchy stress tensor resulting from the elastic law given $\bF_\mathrm{ss}$. Then, if the potential is polyconvex (or rank-one convex), $\tau$ is monotonically increasing in $\gamma$ \parencite[Prop.\ 5.13]{Wollner2026a}. In the following, we employ the material parameters\footnote{For ease of exposition, we omit any units.}
\begin{equation}\label{eq:mat_param}
(a,\,b,\,c)=(-5,\,-14,\,-22)\,,
\end{equation}
for which the true shear stress given \eqref{eq: strain-energy function} can be calculated as \parencite{horgan2010}
\begin{equation}\label{eq:tau_model}
\begin{split}
\tau(\gamma) &= \biggl(\frac{\nu_1^2}{\nu_1^2+1}\,\mathrm{D}_{\nu_1}{\psi}-\frac{\nu_2^2}{\nu_2^2+1}\,\mathrm{D}_{\nu_2}{\psi}\biggr)\Bigg|_{\begin{subarray}{l}\nu_1=\big(\gamma+\sqrt{4+\gamma^2}\big)/2
\\\nu_2=2/\big(\gamma+\sqrt{4+\gamma^2}\big)\\\nu_3=1\end{subarray}} \\
&=\frac{19\,\gamma}{\sqrt{4+\gamma^2}}\biggl(\SM\Bigl(19\sqrt{4+\gamma^2}
-41\Bigr)+\SM\Bigl(19\sqrt{4+\gamma^2}+41\Bigl)-1\biggr)
\\
&\hphantom{=}\:
+9\Bigl(\SM\bigl(9\,\gamma-3\bigr)+\SM\bigl(9\,\gamma+3\bigl)-1\Bigr)\,,
\end{split}
\end{equation}
where the sigmoid function $\SM(x)=\mathrm{D}_x\SP(x)=\exp x/(1+\exp x)$ denotes the first derivative of the softplus function, visualized in in Fig.\ \ref{fig:softplus}(b). As expected, the true shear stress of \eqref{eq: strain-energy function} is monotonic, which we visualize in Fig.~\ref{fig:uniax_example}(a). Notably, the shear stress has an asymptote at $\tau=28$, which can be seen by evaluating \eqref{eq:tau_model} for large $\gamma$, where $\SM(x)\rightarrow 1$ as $x\rightarrow\infty$, see also Fig.\ \ref{fig:softplus}(b).
\begin{theorem}\label{th:non_tsts}
The potential $\psi$ in~\eqref{eq: strain-energy function} violates TSTS-M in uniaxial tension for the material parameters defined in \eqref{eq:mat_param}.
\end{theorem}
\begin{proof}
For uniaxial tension along $\boldsymbol{e}_1$ of an underlying Cartesian coordinate system, the deformation gradient $\bF_\mathrm{ux}$ and the Cauchy stress $\bsigma_\mathrm{ux}$ are given by
\begin{equation}
[\bF_\mathrm{ux}]=\diag\bigl(\lambda,\,\lambda^{-1/2},\,\lambda^{-1/2}\bigr)\quad\text{and}\quad[\bsigma_\mathrm{ux}]=\diag(\sigma,\,0,\,0)\,,
\end{equation}
where $\lambda>1$ and $\sigma$ denote the stretch and Cauchy stress in tensile direction, respectively. Then, if the TSTS-M condition is fulfilled, $\sigma$ is strictly monotonically increasing in $\lambda$ \parencite{NEFF2025105033}. The Cauchy stress of the potential in \eqref{eq: strain-energy function} can be calculated in closed form, such that
\begin{equation}\label{eq:cauchy_model}
\begin{split}
\sigma(\lambda) &= \biggl(\lambda\,\mathrm{D}_{\nu_1}\widetilde{\psi}-\frac{1}{\sqrt{\lambda}}\,\mathrm{D}_{\nu_3}\widetilde{\psi}\biggr)\Bigg|_{\begin{subarray}{l}\nu_1=\lambda\\\nu_2=\lambda^{-1/2}\\\nu_3=\lambda^{-1/2}\end{subarray}} \\
&=\Bigl(10\lambda-\frac{28}{\lambda}\Bigr)\,\SM\Bigl(5\lambda +\frac{14}{\lambda}-22\Bigr)
\\
&\hphantom{=}\:-\frac{1}{\lambda}\Bigl(\sqrt{\lambda}+1\Bigr)\Bigl(5\sqrt{\lambda}-14\Bigr)\Bigl(\lambda-\sqrt{\lambda}+1\Bigr)\,\SM\biggl(-5\lambda-\frac{14}{\lambda}-22+\frac{10}{\sqrt{\lambda}}+28\sqrt{\lambda}\biggr)
\\
&\hphantom{=}\:-\frac{1}{\lambda}\Bigl(\sqrt{\lambda}-1\Bigr)\Bigl(5\sqrt{\lambda}+14\Bigr)\Bigl(\lambda+\sqrt{\lambda}+1\Bigr)\,\SM\biggl(-5\lambda-\frac{14}{\lambda}-22-\frac{10}{\sqrt{\lambda}}-28\sqrt{\lambda}\biggr)\,.
\end{split}
\end{equation}
A numerical analysis reveals that the Cauchy stress of \eqref{eq: strain-energy function} is not monotonically increasing for some stretches, e.g.,
\begin{equation}
    \sigma(1.5)\approx 14.5 >  \sigma(2.5)\approx 12.3\,.
\end{equation}
Hence, TSTS-M is not fulfilled. 

Alternatively, a loss of TSTS-M can be shown by inspecting the potential \eqref{eq: strain-energy function} through the reduced representation $\widehat{W}_\mathrm{red}^\mathrm{inc}$ in \eqref{eq:pot_reduced} for $\log\lambda_1 = \log\lambda$ and $\log\lambda_2 = -\log\lambda/2$ which reveals a lack of convexity. We provide a visualization of both approaches in Fig.~\ref{fig:uniax_example}(b,c). As expected, the conditions are violated for the same interval of stretch values $\lambda$. 
\end{proof}
\begin{corollary}\label{corr:pc}
Polyconvexity does not imply TSTS-M in the incompressible three-dimensional case. This follows from our counterexample of a potential that is polyconvex but does not satisfsy TSTS-M, cf.\ Theorems~\ref{th:pc} and~\ref{th:non_tsts}.
\end{corollary}
\begin{corollary}\label{corr:rank}
Rank-one convexity does not imply TSTS-M in the incompressible three-dimensional case. This follows from Corollary~\ref{corr:pc} and the observation that polyconvexity implies rank-one convexity.
\end{corollary}
%%%%%%%%%%%%%
%	Conclusion
%%%%%%%%%%%%%
\section{Conclusion}

We recently showed that Ball's sufficient conditions for polyconvexity imply TSTS-M in the incompressible three-dimensional case \parencite[Prop.~3.3]{Wollner2026b}, which raised the question whether this implication also holds for general polyconvex parameterizations \parencite[Sect.~3.4]{Wollner2026b}. As shown in the present work, the answer is no: In the three-dimensional incompressible case, polyconvexity does not imply TSTS-M and polyconvexity alone is not sufficient to guarantee a physically reasonable response for idealized elastic materials.

% \nolinenumbers

%========================================================================
%
\newpage
% \paragraph{CRediT authorship contribution statement}
\noindent\textbf{CRediT authorship contribution statement}\quad
\textbf{Dominik K.\ Klein:} Conceptualization, Formal analysis, Visualization, Writing -- original draft, Writing -- review and editing.
\textbf{Maximilian P.\ Wollner:} Conceptualization, Formal analysis, Writing -- review and editing. 
\textbf{Patrizio Neff:} Conceptualization, Writing -- original draft, Writing -- review and editing.
\\[0.75\baselineskip]
% \paragraph{Conflict of interest}
\noindent\textbf{Conflict of interest}\quad
The authors declare that they have no conflict of interest.
\\[0.75\baselineskip]
% \paragraph{Acknowledgments}
\noindent\textbf{{Acknowledgments}}\quad
We want to acknowledge that, independently of our work and roughly at the same time, Gian-Luca~Geuken found a similar counterexample. Moreover, we want to thank him and David~Wiedemann for helpful comments on the manuscript. Dominik\ K.\ Klein acknowledges the financial support provided by the Deutsche Forschungsgemeinschaft (DFG, German Research Foundation, project number 492770117) and the Graduate School of Computational Engineering at TU Darmstadt.
\\[0.75\baselineskip]
% \paragraph{Data availability}
\noindent\textbf{Data availability}\quad
The authors have no data to share.
%
%%%%%%%%%%%%%%%%%%%%%%%%%%%%%%%%
% \appendix
% \numberwithin{equation}{section} 
%%%%%%%%%%%%%%%%%%%%%%%%%%%%%%%%
\vspace{-0.5\baselineskip}
\renewcommand*{\bibfont}{\footnotesize}
\printbibliography

@article{horgan2010,
    author = {C. O. Horgan and J. G. Murphy},
    title = {Simple shearing of incompressible and slightly compressible isotropic nonlinearly elastic materials},
    journal = {J. Elast.},
    year = {2010},
number={98},
pages={205--221},
doi={10.1007/s10659-009-9225-1}
}

@article{Schroeder_Neff_Balzani_2005, 
title={A variational approach for materially stable anisotropic hyperelasticity}, 
volume={42},
ISSN={00207683},
DOI={10.1016/j.ijsolstr.2004.11.021}, 
number={15}, 
journal={Int. J. Solids Struct.}, 
author={Schröder, J. and Neff, P. and Balzani, D.}, 
year={2005}, 
pages={4352--4371}
}

@article{Zee1983,
  title={Ordinary and strong ellipticity in the equilibrium theory of incompressible hyperelastic solids},
  author={L. Zee and E. R. Sternberg},
  journal={Arch. Ration. Mech. Anal.},
  year={1983},
  volume={83},
  pages={53-90},
  doi={10.1007/BF00281087}
}

@article{Baaser2026,
    author  = {Baaser, H.},
    title   = {Hyperelastic stability landscape: a check for {HILL} stability of isotropic, incompressible hyperelasticity depending on material parameters},
    year    = {2026},
    journal = {J. Elast.},
    volume  = {158},
    number  = {8},
    pages   = {1--34},
    doi     = {10.1007/s10659-025-10183-z}
}

@article{Hill1970,
    title   = {Constitutive inequalities for isotropic elastic solids under finite strain},
    author  = {Hill, R.},
    journal = {J. Mech. Phys. Solids},
    volume  = {314},
    pages   = {457--472},
    year    = {1970},
    doi     = {10.1098/rspa.1970.0018}
}

@article{Ball1976,
    author  = {Ball, J. M.},
    title   = {Convexity conditions and existence theorems in nonlinear elasticity},
    year    = {1976},
    journal = {Arch. Rational Mech. Anal.},
    volume  = {63},
    pages   = {337--403},
    doi     = {10.1007/BF00279992}
}

@incollection{Ball1977,
  author     = {Ball, J. M.},
  editor     = {Knops, R. J.},
  title      = {Constitutive inequalities and existence theorems in nonlinear elastostatics},
  booktitle  = {Nonlinear analysis and mechanics: {Heriot-Watt Symposium}},
  publisher  = {Pitman Publishing},
  volume     = {1},
  pages      = {187--241},
  year       = {1977},
}

@book{Ciarlet1988,
    author    = {Ciarlet, P. G.},
    title     = {Mathematical Elasticity Volume I: Three-Dimensional Elasticity},
    year      = {1988},
    series    = {Studies in Mathematics and its Applications},
    publisher = {North-Holland Publishing Company},
}

@article{Dacorogna1993,
author = {Dacorogna, Bernard and Koshigoe, Hideyuki},
journal = {Ann. Fac. Sci. Toulouse Math.},
number = {2},
pages = {163--184},
title = {On the different notions of convexity for rotationally invariant functions},
volume = {2},
year = {1993}
}

@article{Klein2026a,
    author  = {Klein, D. K. and Mokarram, H. and Kikinov, K. and Kannapinn, M. and Rudykh, S. and Gil, A. J.},
    title   = {Neural networks meet hyperelasticity: {A} monotonic approach},
    year    = {2026},
    journal = {Eur. J. Mech., A/Solids},
    volume  = {116},
    pages   = {105900},
    doi     = {10.1016/j.euromechsol.2025.105900}
}

@article{Neff2025b,
    author  = {Neff, P. and Husemann, N. J. and Korobeynikov, S. N. and Ghiba, I.-D. and Martin, R. J.},
    title   = {A natural requirement for objective corotational rates---on structure-preserving corotational rates},
    journal = {Acta. Mech.},
    pages   = {2657--2689},
    volume  = {236},
    year    = {2025},
    doi     = {10.1007/s00707-025-04249-1}
}

@article{NEFF2025105033,
title = {The corotational stability postulate: Positive incremental Cauchy stress moduli for diagonal, homogeneous deformations in isotropic nonlinear elasticity},
journal = {Int. J. Non-Linear Mech.},
volume = {174},
pages = {105033},
year = {2025},
issn = {0020-7462},
doi = {10.1016/j.ijnonlinmec.2025.105033},
author = {Patrizio Neff and Nina J. Husemann and Aurélien S. {Nguetcho Tchakoutio} and Sergey N. Korobeynikov and Robert J. Martin},
}

@article{Mielke2005,
    author  = {Mielke, A.},
    title   = {Necessary and sufficient conditions for polyconvexity of isotropic functions},
    year    = {2005},
    journal = {J. Convex Anal.},
    volume  = {12},
    number  = {2},
    pages   = {291--314},
}

@Inbook{Dunn2003,
author="Dunn, J. Ernest
and Fosdick, Roger
and Zhang, Ying",
editor="Podio-Guidugli, P.
and Brocato, M.",
title="Rank 1 convexity for a class of incompressible elastic materials",
bookTitle="Rational Continua, Classical and New: A collection of papers dedicated to Gianfranco Capriz on the occasion of his 75th birthday",
year="2003",
publisher="Springer Milan",
address="Milano",
pages="89--96",
isbn="978-88-470-2231-7",
doi="10.1007/978-88-470-2231-7_7",
url="https://doi.org/10.1007/978-88-470-2231-7_7"
}

@article{Ghiba2026b,
    author  = {Ghiba, I.-D. and Wollner, M. P. and Neff, P.},
    title   = {Polyconvexity implies {H}ill’s inequality in {SL(2)}},
    note    = {(in preparation)},
    year    = {n.a.},
    journal = {n.a}
}

@article{Neff2016a,
    author  = {Neff, P. and Eidel, B. and Martin, R. J.},
    title   = {Geometry of logarithmic strain measures in solid mechanics},
    journal = {Arch. Rational Mech. Anal.},
    pages   = {507--572},
    volume  = {222},
    year    = {2016},
    doi     = {10.1007/s00205-016-1007-x}
}

@article{Neff2025c,
    author  = {Neff, P. and Holthausen, S. and d'Agostino, M. V. and Bernardini, D. and Sky, A. and Ghiba, I.-D. and Martin, R. J.},
    title   = {Hypo-elasticity, {{C}auchy}-elasticity, corotational stability and monotonicity in the logarithmic strain},
    journal = {J. Mech. Phys. Solids},
    pages   = {106074},
    volume  = {202},
    year    = {2025},
    doi     = {10.1016/j.jmps.2025.106074}
}

@article{Neff2026a,
    author  = {Neff, P. and Husemann, N. J. and Holthausen, S. and Gmeineder, F. and Blesgen, T.},
    title   = {Rate-form equilibrium for an isotropic {C}auchy-elastic formulation.\\{P}art {I}: {M}odeling},
    journal = {J. Nonlinear Sci.},
    year    = {2026},
    volume  = {36},
    number  = {8},
    pages   = {1--46},
    doi     = {10.1007/s00332-025-10214-y},
}

@article{Truesdell1956,
  title   = {{Das ungelöste Problem der endlichen Elastizitätstheorie}},
  author  = {Truesdell, C.},
  year    = {1956},
  journal = {Z. angew. Math. Mech.},
  volume  = {36},
  number  = {3},
  pages   = {97--103}
}

@article{richter1949,
  title   = {Verzerrungstensor, Verzerrungsdeviator und Spannungstensor bei endlichen Formänderungen},
  author  = {H. Richter},
  year    = {1949},
  journal = {Z. angew. Math. Mech.},
  volume  = {29},
  number  = {3},
  pages   = {65--75}
}

@article{Wiedemann2026,
    author  = {Wiedemann, D. and Peter, M. A.},
    title   = {Characterization of polyconvex isotropic functions},
    year    = {2026},
    journal = {Calc. Var.},
    volume  = {65},
    pages   = {115},
    doi     = {10.1007/s00526-025-03222-z}
}

@article{Wollner2026a,
    author  = {Wollner, M. P. and Holzapfel, G. A. and Neff, P.},
    title   = {In search of constitutive conditions in isotropic hyperelasticity: polyconvexity versus true-stress-true-strain monotonicity},
    year    = {2026},
    journal = {J. Mech. Phys. Solids},
    volume  = {209},
    pages   = {106465},
    doi     = {10.1016/j.jmps.2025.106465}
}

@article{Wollner2026b,
    author  = {M. P. Wollner and D. K. Klein and H. Baaser and G. A. Holzapfel and P. Neff},
    title   = {Concurrent enforcement of polyconvexity and true-stress-true-strain monotonicity in incompressible isotropic hyperelasticity: {a}pplication to neural network constitutive models},
    year    = {2026},
journal={Pre-print under review},
      eprint={2605.20031},
      archivePrefix={arXiv},
}

@article{klein2026limitationspolyconvexity,
      title={On limitations of polyconvexity}, 
      author={Dominik K. Klein and Rogelio Ortigosa and Heinrich T. Roth and Karl A. Kalina and Jesús Martínez-Frutos and Markus Kästner and Oliver Weeger},
      year={2026},
      eprint={2605.31392},
      archivePrefix={arXiv},
journal={Pre-print under review}
}

@article{Geuken2026,
    author  = {Geuken, G.-L. and Kurzeja, P. and Wiedemann, D. and Zlatić, M. and Čanađija, M. and Mosler, J.},
    title   = {Modeling isotropic polyconvex hyperelasticity by neural networks -- sufficient and necessary criteria for compressible and incompressible materials},
    year    = {2026},
      archivePrefix={arXiv},
journal={Pre-print under review},
    eprint  = {2603.27351},
}

@misc{neff2019axiomaticintroductionarbitrarystrain,
      title={The axiomatic introduction of arbitrary strain tensors by Hans Richter -- a commented translation of "Strain tensor, strain deviator and stress tensor for finite deformations"}, 
      author={Patrizio Neff and Kai Graban and Eva Schweickert and Robert J. Martin},
      year={2019},
      archivePrefix={arXiv},
journal={Pre-print under review},
      eprint={1909.05998},
}

@article{Rosakis1997,
    author  = {Rosakis, P.},
    title   = {Characterization of convex isotropic functions},
    year    = {1997},
    journal = {	J. Elast.},
    volume  = {49},
    pages   = {257--267},
    doi     = {10.1023/A:1007468902439}
}

@misc{Neff2025,
    author   = {Neff, P. and Husemann, N. J. and Holthausen, S. and d'Agostino, M. V. and Bernardini, D. and Sky, A. and Tchakoutio Nguetcho, A. S. and Ghiba, I.-D. and Martin, R. J. and Gmeineder, F. and Korobeynikov, S. N. and Blesgen, T.},
    title    = {{Truesdell’s Hauptproblem: o}n constitutive stability in idealized isotropic nonlinear elasticity and a 500€ challenge},
    doi 	 = {10.13140/RG.2.2.15349.69603},
    year	 = {2025},
}
\end{document}